\documentclass{article}
\usepackage{geometry}
\usepackage{algorithmic}
\usepackage{textcomp}
\usepackage{cite}
\usepackage{graphicx}
\usepackage{amsmath}
\usepackage{amsthm}
\usepackage{amsfonts}
\usepackage{amssymb}
\usepackage{color}
\usepackage[switch]{lineno}
\newtheorem{theorem}{Theorem}

\newtheorem{corollary}{Corollary}

\newtheorem{definition}{Definition}
\newtheorem{example}{Example}
\newtheorem{lemma}{Lemma}
\newtheorem{notation}{Notation}

\newtheorem{remark}{Remark}
\newcommand{\ra}{\rightarrow}
\newcommand{\p}{\partial}

\DeclareMathOperator*\Lie{Lie}
\newcommand{\red}[1]{{\color{red} {#1}}}
\begin{document}

\title{Further Geometric and Lyapunov Characterizations of Incrementally Stable Systems on Finsler Manifolds}
\author{Dongjun~Wu \thanks{
D. Wu is with Laboratoire des
Signaux et Syst\`emes, CNRS-CentraleSup\'elec-Universit\'e Paris Saclay,
91192 Gif-sur-Yvette, France and the Center for Control Theory and Guidance Technology, Harbin
Institute of Technology, Heilongjiang 150001, China and  (e-mail: dongjun.wu@l2s.centralesupelec.fr;
wdjhit@163.com)}, Guang-ren~Duan~ \thanks{
G. Duan is with the Center for Control Theory and Guidance Technology,
Harbin Institute of Technology, Heilongjiang 150001, China (e-mail:
g.r.duan@hit.edu.cn)}}
\date{}
\maketitle

\begin{abstract}
	In this paper, we report several new geometric and Lyapunov
	characterizations of incrementally stable systems on Finsler and
	Riemannian manifolds. A new and intrinsic proof of an important
	theorem in contraction analysis is given via the complete lift of
	the system. Based on this, two Lyapunov characterizations of
	incrementally stable systems are derived, namely, converse
	contraction theorems, and revelation of the connection between
	incremental stability and stability of an equilibrium point, in which
	the second result recovers and extends the classical Krasovskii's
	method.

	{\color{red} A technical mistake in Theorem 4 has been corrected in this
	version compared to the IEEE TAC version \cite{Wu2021}. See the red texts.}

\end{abstract}


\section{Introduction and Motivation}

Incremental stability, also termed as contraction in the
literature, is concerned with the
attractive behaviour of any pairs of trajectories of
a system. This notion dates back to the 1950s and
1960s, known as extreme stability, which was introduced to study the 
stability of periodic orbits of dynamical systems \cite{lasalle1957study,
salle1961stability}; at the time, the problem was usually tackled by transforming it into a set
	stability analysis problem. More precisely, consider the system $\dot{x}
	= F(x)$ and its copy $\dot{y}=F(y)$. Then the incremental stability of
	the system with state $x$ is equivalent to the stability of the set
	$\{x=y\}$ of the augmented system with state $(x,y)$. More recent
	developments in this line of research can be found in
	\cite{angeli2002lyapunov, ruffer2013convergent} and the references
	therein.
	Another widely adopted approach is related to the properties of the
	differential dynamics of the system $\dot{x}=F(x)$. For
	example, the famous Demidovich condition \cite{demidovich1961,
	demidovich1967lectures} --- proposed in the 1960s --- which involves the
	Jacobian of $F$, can be viewed as a Lyapunov condition imposed on the
	differential dynamics $\delta \dot{x} = \frac{\p F}{\p x} \delta x$.
	However, it was not until the 1990s --- when W. Lohmiller and J.
	Slotline published the paper \cite{lohmiller1998contraction}, suggesting
	studying the differential dynamics of the system --- that this research
	direction attracted new attentions in the control community. Since then,
	this point of view has been extensively studied, enriched, and has been
	applied to solve various control problems, such as synchronization \cite
	{slotine2005study}, \cite{wang2005partial}, trajectory tracking \cite
	{pavlov2008incremental}, \cite{reyes2018virtual}, \cite
	{lohmiller2000nonlinear}, observer design \cite{wang2005partial}, \cite
	{jouffroy2010tutorial}, to name a few. 
 An alternative approach to contraction analys is the matrix measure method
\cite {sontag2010contractive, aminzare2014contraction}, which is applicable
to systems defined on normed vector spaces.



Despite the success of contraction analysis in applications, notably those
inspired by the work of Lohmiller and J. Slotline etc., there had always been a
call for Lyapunov or geometric characterizations of contraction analysis, see
for example \cite{angeli2002lyapunov, forni2013differential}.
In \cite {forni2013differential}, F. Forni \textit{et al.} proposed a
differential Lyapunov framework for contraction analysis, in which they
introduced two essential objects, namely, the Finsler structure and the
Finsler-Lyapunov function to derive sufficient conditions for incremental
stability. The advantage of this framework was illustrated by showing that
numerous previous works in the literature could be unified utilizing a single
condition. Nevertheless, there still remain several important issues and
interesting questions of this framework that need to be addressed: 

\begin{itemize}
\item[\textbf{Q1}] Most of the resutls in \cite{forni2013differential} as
well as their proofs are handled in local coordinates. Therefore, the
geometric interpretations of the differential conditions obtained therein
need to be clarified. This leads to the following question: can we
reformulate all the results in a coordinate free way?

\item[\textbf{Q2}] The main theorem in \cite{forni2013differential}
(Theorem \ref{thm: Forni}) gives a sufficient condition for incremental
stability. A natural question is, is it also necessary?  Or in another word, can
we prove the converse theorem?

\item[\textbf{Q3}] When a system has an equilibrium, then incremental
stability implies certain type of stability of the equilibrium. In this case,
what is the connection between incremental stability and stability of the
equilibrium?
\end{itemize}

We provide answers to all the questions above, which form
into the main contributions of this paper: 
\begin{itemize}
	\item We give an intrinsic form condition expressed in the tangent
		bundle, which 
		guarantees the incremental stability of the system. This is
		achieved by studying the behaviour of the complete lift of the
		system. The result easily recovers one of the main results in
		\cite{forni2013differential} and gives new geometric
		insights to it. 
	\item We prove converse theorems of uniform exponential incremental
		stability and uniform asymptotical incremental stability, in
		a coordinate-free way. The results are expressed in the tangent
		bundle --- differential in nature --- involving no copy of the
		original system (cf. \cite{angeli2002lyapunov}).
	\item We show that the relationship between incremental stability and
		stability is linked by the so called Krasovskii's method. More
		precisely, we prove that a Lyapunov function can be directly
		constructed if we already have a Finsler-Lyapunov function at
		hand. Surprisingly, the answer to Q3 is related to Lie bracket.
\end{itemize}



\begin{notation}
We list some of the notations used in this paper. $\mathcal{X}$: Riemannian
manifold; $T_{x}\mathcal{X}$: tangent space at $x$; $\left\langle
v_{x},u_{x}\right\rangle $: Riemannian product of $v_{x},u_{x}\in T_{x}
\mathcal{X}$; $\ell(c)$: length of the curve $c$; $d(x,y)$: Riemannian
distance between $x$ and $y$; $\mathcal{L}_{f}V(t,x)$: timed Lie derivative
of $V(t,x)$ along the flow of $f(x,t)$, see \cite{wu2020intrinsic}; $P_{p}^{q}$: parallel transport from 
$p$ to $q$; $\phi_{\ast}$: push forward of a diffeomorphism $\phi:\mathcal{M}
\rightarrow\mathcal{N}$; $\phi^{\ast}$:\ pull back of a smooth map $\phi:
\mathcal{M}\rightarrow\mathcal{N}$; $B_{x}^{c}$: the open ball with radius $c
$ centered at $x$; $\phi(t;t_{0},x_{0})$: the flow of a system with initial
condition $(t_{0},x_{0})$; $[f,g]$: the Lie bracket of two smooth vector fields; $\pi:T\mathcal{X}\rightarrow\mathcal{X}$: the
projection map from the tangent bundle to its base space; $O(s^{k})$: $k$
order of $s$.
\end{notation}
\section{Preliminaries}

In \cite{forni2013differential}, the authors have shown that a natural
setting for contraction analysis is Finsler geometry. They introduced the
concept Finsler-Lyapunov function (FLF) which is crucial to the
characterization of incremental stability. Given a Finsler structure $F$
on manifold $\mathcal{X}$, a candidate FLF $V(t,x,\delta x)$ is a
non-negative function defined on the tangent bundle satisfying 
\begin{equation*}
c_{1}F(x,\delta x)^{p}\leq V(t,x,\delta x)\leq c_{2}F\left( x,\delta
x\right) ^{p}, 
\end{equation*}
for all $(x,\delta x)\in T\mathcal{X},$ where $c_{1},c_{2}>0$, $p\geq1$.
Consider the nonlinear time varying system
\begin{equation}
\dot{x}=f(x,t),   \label{NLTV no input}
\end{equation}
evolving on the Riemannian manifold $(\mathcal{X},g)$ where $f(x,t)$ is $
\mathcal{C}^{1}$. In the setting of Riemannian manifold, a candidate FLF should
verify the following condition:
\begin{equation}
c_{1}|\delta x|^{p}\leq V(t,x,\delta x)\leq c_{2}|\delta x|^{p}   \label{FLF}
\end{equation}
where $\left\vert \cdot\right\vert $ denotes the induced norm of the
Riemannian metric, i.e. $\left\vert \delta x\right\vert =\sqrt{\left\langle
\delta x,\delta x\right\rangle }$. The Riemannian distance induced by $g$ is
\begin{equation*}
d(x_{1},x_{2})=\inf_{\gamma\in\Gamma(x_{1},x_{2})}\int_{0}^{1}\left\vert
\gamma^{\prime}(s)\right\vert ds 
\end{equation*}
where $\Gamma(x_{1},x_{2})$ is the set of smooth curves joining $x_{1}$ to $
x_{2}$.

In what follows, we introduce the the definitions of local incremental
stability (IS) and extend it to global IS.

\begin{definition}[Local and Global IS]
\label{def: LIS}The system (\ref{NLTV no input}) is called

\begin{enumerate}
\item \emph{uniformly locally incremental stable (ULIS) at }$x$ if there
exits a class $\mathcal{K}$ function $\alpha$ and a positive constant $
\varepsilon$, such that for all $t\geq t_{0}\geq0,$
\begin{equation}
d(\phi(t;t_{0},x_{1}),\phi(t;t_{0},x_{2}))<\alpha(d(x_{1},x_{2})), 
\label{eq: LIS}
\end{equation}
for all $x_{1},x_{2}\in B_{\varepsilon}(x);$ \emph{uniformly incremental
stable (UIS) on }$D$ if (\ref{eq: LIS}) is satisfied for all $x_{1},x_{2}\in
D$; \emph{uniformly globally incremental stable (UGIS) }if the system is UIS
on $\mathcal{X}$;

\item \emph{uniformly locally incremental asymptotically stable (ULIAS) at }$
x$ if it is ULIS at $x$ and there exists a class $\mathcal{KL}$ function
$ \beta(r,s)$, 
and $\varepsilon$ can be chosen such that for all $t\geq t_{0}\geq0,$
\begin{equation}
d(\phi(t;t_{0},x_{1}),\phi(t;t_{0},x_{2}))\leq\beta(d(x_{1},x_{2}),t-t_{0}) 
\label{eq:LIAS}
\end{equation}
for all $x_{1},x_{2}\in B_{\varepsilon}(x);$ \emph{incremental
asymptotically stable (UIAS) on }$D$ if (\ref{eq:LIAS}) is satisfied for all 
$x_{1},x_{2}\in D$ ; \emph{uniformly globally incremental asymptotically
stable (UGIAS) }if the system is UIAS on $\mathcal{X}$;

\item \emph{uniformly locally incremental exponentially stable (ULIES) at }$x
$ if there exists $K\geq1$, $\lambda>0$ and $\varepsilon>0$ such that for
all $t\geq t_{0}\geq0,$
\begin{equation}
d(\phi(t;t_{0},x_{1}),\phi(t;t_{0},x_{2}))\leq
Ke^{-\lambda(t-t_{0})}d(x_{1},x_{2}),   \label{eq:LIES}
\end{equation}
for all $x_{1},x_{2}\in B_{\varepsilon}(x);$ \emph{uniformly incremental
exponentially stable (UIES) on }$D$ if (\ref{eq:LIES}) is satisfied for all $
x_{1},x_{2}\in D$; \emph{uniformly globally incremental \emph{exponentially}
stable (UGIES) }if the system is UIES on $\mathcal{X}$;
\end{enumerate}
\end{definition}

The following theorem is due to F. Forni \textit{et al.}
\cite{forni2013differential}, which provides a sufficient condition ---
expressed in tangent bundle and local coordinates --- for incremental stability.
\begin{theorem}[F. Forni et. al. \protect\cite{forni2013differential}]
\label{thm: Forni}Consider the system (\ref{NLTV no input}) on $\left( 
\mathcal{X},g\right) $, a connected and forward invariant set $C$, and a
function $\alpha:\mathbb{R}_{+}\rightarrow\mathbb{R}_{+}$. Let $V(x,\delta x)
$ be a candidate FLF such that, in coordinate,
\begin{equation}
 \frac{\partial V}{\partial t}+\frac{\partial V}{
\partial x}f+\frac{\partial V}{\partial\delta x}\frac{
\partial f}{\partial x}\delta x  \label{eq: thm Forni}  \leq-\alpha(V)  \notag
\end{equation}
for each $t\in\mathbb{R}$, $x\in C\subset\mathcal{X}$, and $\delta x\in T_{x}
\mathcal{X}$. Then (\ref{NLTV no input}) is

\begin{itemize}
\item[( IS )] UIS on $C$ if $\alpha(s)=0$ for each $s\geq0;$

\item[(IAS)] UIAS on $C$ if $\alpha(s)$ is a class $\mathcal{K}$ function;

\item[(IES)] UIES on $C$ if $\alpha(s)=\lambda s$ for some $\lambda>0$.
\end{itemize}
\end{theorem}

This paper deals with time varying vector field, so the time varying version
of Lie derivative will be needed. We refer the reader to \cite{wu2020intrinsic} for
its definition.

Several notations of Riemannian geometry will be used in this
paper, such as the geodesic, the exponential map, first variation formula
of arc length and Lipschitz continuous in the Riemannian context etc. These can
be found in \cite{carmo1992riemannian} and \cite{wu2020intrinsic} and the
references therein. We assume that the Riemannian manifolds treated in this
paper are complete, which implies the existence of minimizing geodesic between
any two points on the manifold. Besides, we assume all the geodesic to be
$\mathcal{C }^{1}$. The solutions of (\ref{NLTV no input}) are assumed to be
forward complete.

\section{Complete lift and intrinsic proof of Theorem\label{sec: Clift} 
\protect\ref{thm: Forni}\protect\cite{forni2013differential}}

The main results in \cite{forni2013differential} are essentially local since
the conditions are represented in local coordindates.
This poses the following problem: assume that the manifold $
\mathcal{X}$ is covered by three coordinate charts $U_1, U_2, U_3$,
and $U_1 \cap U_2 = \emptyset$. If we have already known that the system is
UGIAS, then neither $U_1$ nor $U_2$ can be invariant otherwise $\phi(t;t_0,x_1)$
and $\phi(t;t_0,x_2)$ with $x_1 \in U_1$, $x_2 \in U_2$ cannot converge to each
other, contradicting with the UGIAS of the system. Since $U_1$ and $U_2$ are not
invariant sets, Theorem \ref{thm: Forni} cannot be applied on the two sets.
Therefore in order to analyze UGIAS, further analysis will be needed.

To overcome this difficulty, we give an intrinsic proof of Theorem \ref{thm:
Forni} \cite{forni2013differential}. In particular, an intrinsic form
of (\ref{eq: thm Forni}) will be given. The key ingredients we need to
achieve this goal are two concepts from differential geometry: the Lie
transport of a vector and the complete lift of a vector field.

\begin{definition}[Lie Transport\protect\cite{abraham2012manifolds}]
Consider the ordinary differential equation (\ref{NLTV no input}) and its
flow $\phi(s;t,x)$, i.e. the solution to the following Cauchy problem
\begin{equation*}
\frac{d}{dt}\phi(t;s,x)  =f(\phi(t;s,x),t), \;  \phi(s;s,x)  =x
\end{equation*}
for $x\in\mathcal{X}$, $t\geq s\geq0$. The Lie transport of a vector $W\in
T_{x}\mathcal{X}$ along the flow of (\ref{NLTV no input}) is defined as the
push forward of $W$ by $\phi(t;s,x)$ from $T_{x}\mathcal{X}$ to $T_{\phi
(t;s,x)}\mathcal{X}$, i.e. $\text{Lie}(W)(t,s)=\phi(t;s,x)_{\ast}W.$
Thus Lie$(W)$ defines a vector field along the curve $\sigma(t,s)=\phi(t;s,x)
$.
\end{definition}

\begin{definition}[Complete Lift \protect\cite{yano1973tangent}]
\label{def: CLift}Consider the time varying vector field $f(x,t)$. Given a
point $v\in T\mathcal{X}$, let $\sigma(t,s)$ be the integral curve of $f$
with $\sigma(s,s)=\pi(v)$. Let $V(t)$ be the vector field along $\sigma$
obtained by Lie transport of $v$ by $f$. Then $(\sigma,V)$ defines a curve
in $T\mathcal{X}$ through $v$. For every $t\geq s$, the \emph{complete lift}
of $f$ into $TT\mathcal{X}$ is defined at $v$ as the tangent vector to the
curve $(\sigma,V)$ at $t=s$. We denote this vector field by $\tilde{f}(v,t)$
, for $v\in T\mathcal{X}$.
\end{definition}

Complete lift is a term widely used in differential geometry \cite
{yano1973tangent}, \cite{crampin1986applicable}. Its use can also be found
in control theory. A. Schaft \textit{et al.} used this concept to study
prolonged system and differential passivity \cite{cortes2005characterization}
, \cite{van2015geometric}, \cite{van2013differential} in a coordinate free
manner. F. Bullo \textit{et al. }used it to study the linearization of
nonlinear mechanical systems, see \cite{bullo2019geometric} and its online
supplementary materials. Having the two definitions at hand, we can prove
the following coordinate free form of Theorem \ref{thm: Forni}.

\begin{theorem}
\label{thm: Clift}Consider the system (\ref{NLTV no input}), a function $
\alpha:\mathbb{R}_{+}\rightarrow\mathbb{R}_{+}$ and the dynamical system
defined by the complete lift of $f$,
\begin{equation}
\dot{v}=\tilde{f}(v,t),\ v\in T\mathcal{X}.   \label{CLift sys}
\end{equation}
Let $V$ be a candidate FLF satisfying (\ref{FLF}). If
\begin{equation}
\mathcal{L}_{\tilde{f}}V(t,v)\leq-\alpha(V(t,v))   \label{dot V}
\end{equation}
along the system trajectory for $t\in\mathbb{R}_{+}$, $v\in T\mathcal{X}$.
Then (\ref{NLTV no input}) is

\begin{itemize}
\item[( IS )] UGIS if $\alpha(s)=0$ for each $s\geq0$;

\item[(IAS)] UGIAS if $\alpha(s)$ is a class $\mathcal{K}$ function;

\item[(IES)] UGIES if $\alpha(s)=\lambda s$ for some $\lambda>0$.
\end{itemize}
\end{theorem}

\begin{proof}
We only prove the third item, since the first two are similar. By Definition 
\ref{def: CLift}, the trajectory of the system (\ref{CLift sys}) started
from $v$ is the Lie transport of the vector $v$ along $\phi(t;t_{0},\pi(v))$
. Given two points $x_{1}$, $x_{2}$, there is a minimizing geodesic curve
$\gamma :[0,1]\rightarrow\mathcal{X}$ joining $x_{1}$ to $x_{2}$. The following
expression defines a curve in $T\mathcal{X}$: 
\begin{equation*}
t\mapsto\left( \phi(t;t_{0},\gamma(s)),\frac{\partial}{\partial s}
\phi(t;t_{0},\gamma(s))\right) \in T\mathcal{X}
\end{equation*}
which is the Lie transport of the vector $\gamma^{\prime}(s)$ along the
curve $\phi(t;t_{0},\gamma(s))$ for $s\in\lbrack0,1]$ and hence is the
solution to (\ref{CLift sys}). By (\ref{dot V}), the FLF decreases
exponentially along the trajectory of (\ref{CLift sys}), i.e.
\begin{equation*}
V(\phi(t;t_{0},\pi(v)),\text{Lie}(v)(t;t_{0}))\leq V(\pi(v),v)e^{-\lambda
(t-t_{0})}, 
\end{equation*}
therefore 
\begin{align*}
	d(\phi(t &;t_{0},x_{1}),\phi(t;t_{0},x_{2})) \leq\int_{0}^{1}\left\vert \frac{\partial}{\partial s}\phi(t;t_{0},
\gamma(s))\right\vert \text{d}s \\
& \leq c_{1}^{-\frac{1}{p}}\int_{0}^{1}V\left( \phi(t;t_{0},\gamma (s)),
\frac{\partial}{\partial s}\phi(t;t_{0},\gamma(s))\right) ^{\frac{1}{p}}
\text{d}s \\
& =c_{1}^{-\frac{1}{p}}\int_{0}^{1}V\left( \phi(t;t_{0},\gamma (s)),\text{Lie
}(\gamma^{\prime}(s))(t;t_{0})\right) ^{\frac{1}{p}}\text{d}s \\
& \leq c_{1}^{-\frac{1}{p}}\int_{0}^{1}c_{2}^{\frac{1}{p}}\left\vert
\gamma^{\prime}(s)\right\vert e^{-\frac{\lambda}{p}(t-t_{0})}\text{d}s \\
& =\left( \frac{c_{2}}{c_{1}}\right) ^{\frac{1}{p}}e^{-\frac{\lambda}{p}
(t-t_{0})}d(x_{1},x_{2})
\end{align*}
This completes the proof.
\end{proof}

\begin{remark}
We remark that the above proof can be easily adapted to prove the other
theorems in \cite{forni2013differential}.
	As will be seen in the following, the complete lift technique will be
	used throughout this paper, in particular, to prove converse theorems
	and reveal the connection between incremental stability and
	Lyapunov stability of an equilibrium.
	Thus we underscore that this section is crucial for the understanding of
	subsequent results.
\end{remark}

To see that Theorem \ref{thm: Clift} is indeed the intrinsic form of
Theorem \ref{thm: Forni} \cite{forni2013differential}, we just need the
following lemma \cite{crampin1986applicable}.

\begin{lemma}
	\label{lem: CLift}Suppose that $T\mathcal{X}$ has local
	coordinates $\{x,v\}$
and $TT\mathcal{X}$ is locally spanned by $\left\{ \frac{\partial}{\partial
x^i}, \frac{\partial}{\partial v_i}\right\}_{i=1, \cdots, n} $, where $n$ is
the dimension of the manifold $\mathcal{X}$. Then in this coordinate, $
\tilde{f}$ reads 
$$\tilde{f}=
\begin{bmatrix}
f \\ 
(\partial f/\partial x)v
\end{bmatrix}
.$$ 
\end{lemma}

\begin{remark}
In \cite{simpson2014contraction}, the authors gave a coordinate free proof
of a contraction theorem on Riemannian manifold. But we should notice that
there are several differences between our result and theirs. Firstly, in 
\cite{simpson2014contraction}, the function $\left\langle \left\langle
v_{x},v_{x}\right\rangle \right\rangle $ considered in \cite
{simpson2014contraction} is a special case of the more general FLF
considered here. Second, the proof in \cite{simpson2014contraction} relies
on the Levi-Civita connection defined on the Riemannian manifold. In this
paper however, we do not use certain connection on the manifold. Therefore,
our proof can be extended smoothly to Finsler manifold without considering
any connections, this greatly simplifies the analysis. Lastly, using
complete lift to treat this problem is new. The advantage is that it can be
easily modified to prove other results concerning incremental stability.
\end{remark}

\section{Converse Contraction Theorems}

In this section, we prove that the condition proposed by F. Forni
\textit{et al.} to ensure contractive properties is not only sufficient
but also necessary.  That is, if the system has certain incremental stability,
then we should be able to find a FLF. In \cite{angeli2002lyapunov}, D. Angeli
gave a necessary and sufficient conditions of GIAS using the incremental
Lyapunov function, which is a set version of Lyapunov function. In comparison,
what we are going to prove here is a differential version. In
\cite{van2013differential}, V. Andrieu \textit{et al.} proved a converse theorem
for UIES systems defined on $\mathbb{R}^{n}$, see also
\cite{kawano2019contraction} for monotone systems. We postpone the discussion on
the differences between these results after the proof.

In order to streamline the main underlying idea, we assume the system to be
globally stable. Extension to local version is not difficult.
\begin{definition}
\label{def: Lip}A vector field $X$ on $\mathcal{X}$ is said to be globally
Lipschitz continuous on $\mathcal{X}$, if there exists a constant $L>0$ such
that for $p,q\in\mathcal{X}$ and all $\gamma$ geodesic joining $p$ to $q$,
there holds
\begin{equation*}
\left\vert P_{p}^{q}X(p)-X(q)\right\vert \leq Ld(p,q). 
\end{equation*}
\end{definition}
\begin{remark}
	This condition can also be replaced with bound of the covariant
	derivative of $X$, see \cite{wu2020geometric}.
\end{remark}
\begin{theorem}
\label{thm: converse}Consider the system (\ref{NLTV no input}) defined on
Riemannian manifold $(\mathcal{X},g)$ with $f\in\mathcal{C}^{2}$ and global
Lipschitz continuous with constant $L$ in the sense of Definition \ref{def:
Lip}. Then the system is UGIES if and only if there exists a (possibly time
dependent) $\mathcal{C}^{0}$ FLF, $V(t,v):\mathbb{R}_{+}\times T\mathcal{
X\rightarrow}\mathbb{R}$ such that for any $p\geq1$

\begin{enumerate}
\item There exists two positive constants $c_{1},c_{2}$, such that
\begin{equation*}
c_{1}|v|^{p}\leq V(t,v)\leq c_{2}|v|^{p},\ \forall v\in T\mathcal{X}
\end{equation*}
where $|\cdot|$ is the norm induced by the Riemannian metric.

\item The timed Lie derivative of $V(t,v)$ along the system (\ref{CLift sys}
) satisfies 
\begin{equation}
\mathcal{L}_{\tilde{f}}V(t,v)\leq-kV(t,v)   \label{thmeq:1}
\end{equation}
where $\mathcal{L}_{\tilde{f}}$ is the timed Lie derivative along the flow
of $\tilde{f}$.
\end{enumerate}

Additionally, when the map $v\mapsto|v|^{p}:$ $T\mathcal{X}\rightarrow 
\mathbb{R}$ is $\mathcal{C}^{1}$ (for example when $p$ is an even number), $
V(t,v)$ is also $\mathcal{C}^{1}$.
\end{theorem}

To prove the theorem, we need some basic tools from Riemannian geometry and
a few lemmas that we are going to derive. Part of these materials can be
found in our previous work \cite{wu2020geometric}. We remark that, in general,
the FLF $V$ is time dependent, in contrast to the time independent version in 
\cite{forni2013differential}.

The following lemmas are key to the proof.

\begin{lemma}
	\label{lem: esti of Lip}Given that the system (\ref{NLTV no
		input}) is
globally Lipschitz continuous with constant $L$ in the sense of Definition 
\ref{def: Lip}, then there holds the following estimation
\begin{align}
d(x_{1},x_{2})e^{-L(\tau-t)} & \leq d(\phi(\tau;t,x_{1})),\phi(\tau
;t,x_{2}))  \label{esti of Lip} \\
& \leq d(x_{1},x_{2})e^{L(\tau-t)}, \forall\tau\geq t,\ \forall x\in\mathcal{X}.
\end{align}
\end{lemma}

\begin{lemma}
\label{lem: dist}Given $\gamma_{1},\gamma_{2}\in\mathcal{C}^{1}(\mathbb{R}
_{+};\mathcal{X})$, where $\mathcal{X}$ is a Riemannian manifold. If 
$\gamma_{1}(0)=\gamma_{2}(0)=x$
and $\gamma_{1}^{\prime}(0)=\gamma_{2}^{\prime}(0)=v$,
then $d(\gamma_{1}(s),\gamma_{2}(s))=O(s^{2})$
when $s$ is sufficiently small.
\end{lemma}

The proof of the two lemmas can be found in \cite{wu2020geometric}.

\begin{lemma}
\label{lem: upp bd Lie(v)}If the system (\ref{NLTV no input}) is UGIES, i.e.
\begin{equation}
d(\phi(t;t_{0},x_{1}),\phi(t;t_{0},x_{2}))\leq
ke^{-\lambda(t-t_{0})}d(x_{1},x_{2}),   \label{eq: GIES}
\end{equation}
for any $x_{1},x_{2}\in\mathcal{X}$, then the Lie transport of any
trajectory of the system (\ref{NLTV no input}) satisfies 
\begin{equation}
\left\vert \text{Lie}(v)(t,t_{0})\right\vert \leq
ke^{-\lambda(t-t_{0})}|v|,\ \forall v\in T\mathcal{X}.   \label{eq:0}
\end{equation}
\end{lemma}

\begin{proof}
Denote the normalized geodesic joining $x_{1}$ to $x_{2}$ as $\gamma :[0,
\hat{s}]\rightarrow\mathcal{X}$. We have $0\leq\hat{s}=d(x_{1},x_{2})$. Let $
v\in T\mathcal{X}$ with $\pi_{T\mathcal{X}}(v)=x_{1}$ and $v=\gamma
^{\prime}(0)$. Denote $v_{t}=$Lie$(v)(t,t_{0})$, we have

\begin{equation}
\hat{s}\left\vert v_{t}\right\vert =d\left( \exp_{\phi(t;t_{0},x_{1})}\left( 
\hat{s}v_{t}\right) ,\phi\left( t;t_{0},x_{1}\right) \right) ,   \label{eq:1}
\end{equation}
where $\exp_{x}:T\mathcal{X}\rightarrow\mathcal{X}$ is the exponential map.
Since we have assumed that the Riemannian manifold is complete, $\exp_{x}$
is defined on $T\mathcal{X}$ for all $x\in\mathcal{X}$. Using the metric
property of $d$, we have
\begin{align}
 d&\left(  \exp_{\phi(t;t_{0},x_{1})}\left( \hat{s}v_{t}\right) ,\phi\left(
		t;t_{0},x_{1}\right) \right)  \notag \\
	& \leq d\left( \exp_{\phi(t;t_{0},x_{1})}\left( \hat{s}v_{t}\right)
		,\phi(t;t_{0},x_{2})\right)  \notag \\
	& \;\; +d(\phi(t;t_{0},x_{2}),\phi(t;t_{0},x_{1})) \\
	& \leq d\left( \exp_{\phi(t;t_{0},x_{1})}\left( \hat{s}v_{t}\right)
		,\phi(t;t_{0},x_{2})\right) +K\hat{s}e^{-\lambda(t-t_{0})},   \label{eq:2}
\end{align}
where the second inequality holds due to (\ref{eq: GIES}). From (\ref
{eq:1}) and (\ref{eq:2}) we get
\begin{equation}
\left\vert v_{t}\right\vert \leq\frac{d\left(
\exp_{\phi(t;t_{0},x_{1})}\left( \hat{s}v_{t}\right)
,\phi(t;t_{0},x_{2})\right) }{\hat{s}}+Ke^{-\lambda(t-t_{0})}.   \label{eq:3}
\end{equation}
\begin{figure}[ptb]
\centering
\includegraphics[scale = 0.27 ]{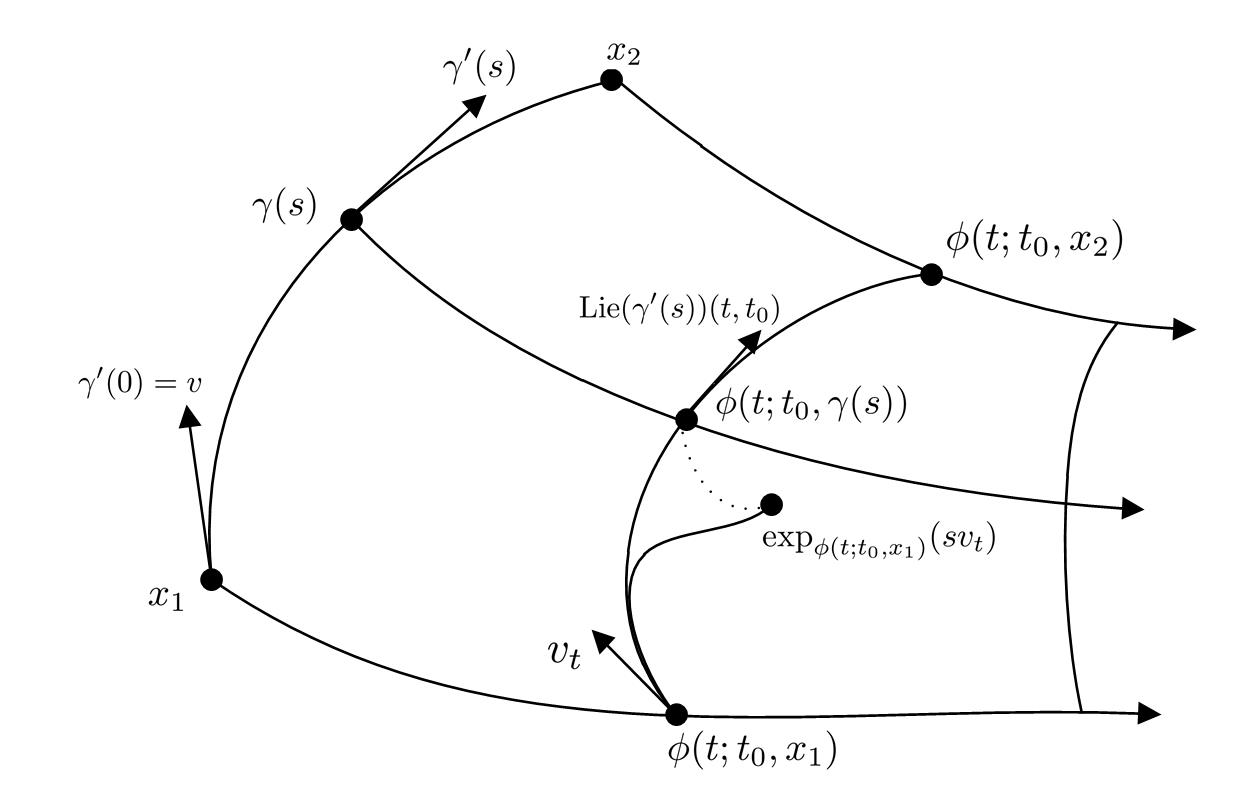}
\caption{Illustration of the proof}
\label{fig: proof}
\end{figure}
See Fig. 1. for an illustration.
Now we want to show that the first term on the right hand side is of order $
O(\hat{s}^{2})$. Since $x_{2}=\gamma(\hat{s})$, the term can also be written
as
\begin{equation*}
\kappa(s)=\frac{d\left( \exp_{\phi(t;t_{0},x_{1})}\left( sv_{t}\right)
,\phi(t;t_{0},\gamma(s))\right) }{s}
\end{equation*}
where we have replaced $\hat{s}$ with $s$. For this, we consider the two
functions $ \alpha_{1}(s)  =\exp_{\phi(t;t_{0},x_{1})}\left( sv_{t}\right)$,
$\alpha_{2}(s)  =\phi(t;t_{0},\gamma(s))$.
We have $\alpha_{1}(0)=\alpha_{2}(0)=x_{1}$ and $\alpha_{1}^{\prime}(0)=
\alpha_{2}^{\prime}(0)=v_{t}.$ Thus
\begin{equation*}
\kappa(s)=d(\alpha_{1}(s),\alpha_{2}(s))/s=O(s) 
\end{equation*}
invoking Lemma \ref{lem: dist}. Now letting $\hat{s}\rightarrow0$ in (
\ref{eq:3}), we obtain (\ref{eq:0}), which completes the proof.
\end{proof}

The lower bound of Lie$(v)(t;t_{0})$ is also needed.

\begin{lemma}
\label{lem: low bd of Lie(v)}Suppose that $f$ in (\ref{NLTV no input}) is
Lipschitz continuous with constant $L$ in the sense of Definition \ref{def:
Lip}, the Lie transport verifies 
$\left\vert \text{Lie}(v)(t;t_{0})\right\vert \geq|v|e^{-L(t-t_{0})},\
\forall v\in T\mathcal{X}.$
\end{lemma}

\begin{proof}
Let $\gamma(s)=\exp_{x}(sv)$, so $\gamma^{\prime}(0)=v$. From Lemma \ref
{lem: esti of Lip}, we have the following inequality for $s>0$:
\begin{equation*}
\frac{d(x,\gamma(s))e^{-L(\tau-t)}}{s}\leq\frac{d(\phi(t;t_{0},x)),\phi
(t;t_{0},\gamma(s)))}{s}, 
\end{equation*}
in which the left hand side is nothing but $|v|e^{-L(t-t_{0})}$. Letting $
s\rightarrow0+$, 
\begin{align*}
 \lim_{s\rightarrow0+} &\frac{d(\phi(t;t_{0},x)),\phi(t;t_{0},\gamma(s)))}{s}
\\
& =\left. \frac{\text{d}}{\text{d}s}\right\vert
_{s=0}d(\phi(t;t_{0},x)),\phi(t;t_{0},\gamma(s))) \\
& =\lim_{s\rightarrow0}\left\langle \frac{\partial\phi(t;t_{0},\gamma (s))}{
\partial s},\bar{\gamma}^{\prime}(\hat{s})\right\rangle \\
& \leq\lim_{s\rightarrow0}\left\vert \frac{\partial\phi(t;t_{0},\gamma (s))}{
\partial s}\right\vert \\
&=\left\vert \text{Lie}(v)(t;t_{0})\right\vert
\end{align*}
where $\bar{\gamma}:[0,\hat{s}]\rightarrow\mathcal{X}$ is the normalized
geodesic joining $\phi(t;t_{0},x)$ to $\phi(t;t_{0},\gamma(s))$. Thus the
proof is completed.
\end{proof}

Now we are in position to prove Theorem \ref{thm: converse}.

\begin{proof}[Proof of Theorem \protect\ref{thm: converse}]
Necessity is already proven in Section \ref{sec: Clift}. 
It remains to prove the converse. 

\textbf{Step 1}:\ We
consider the following candidate FLF:
\begin{equation}
V(t,v)=\int_{t}^{t+\delta}\left\vert \text{Lie}(v)(\tau;t)\right\vert ^{p}
\text{d}\tau   \label{cand FLF}
\end{equation}
where $p\geq1$, $\delta>0$. From Lemma \ref{lem: esti of Lip} and Lemma \ref{lem: upp bd
Lie(v)}, we can estimate the lower and upper bound of $V(t,v)$:
\begin{align*}
V(t,v) & \geq|v|^{p}\int_{t}^{t+\delta}e^{-pL(\tau-t)}\text{d}\tau =\frac{
1-e^{-pL\delta}}{pL}|v|^{p} \\
V(t,v) & \leq|v|^{p}\int_{t}^{t+\delta}e^{-p\lambda(\tau-t)}\text{d}\tau=
\frac{1-e^{-p\lambda\delta}}{p\lambda}|v|^{p}.
\end{align*}
Thus there exists two positive constants $c_{1},c_{2}$ such that
\begin{equation*}
c_{1}|v|^{p}\leq V(t,v)\leq c_{2}|v|^{p}. 
\end{equation*}

\textbf{Step 2}:\ By the property of Lie transport, we know that 
\[
	\Lie(\Lie(v)(t;s))(\tau;t)=\Lie(v)(\tau;s),
\]so 
\begin{align*}
V(s,\text{Lie}(v)(s,t)) & =\int_{s}^{s+\gamma}\left\vert \text{Lie}(\text{Lie
}(v)(s,t))(\tau;s)\right\vert ^{p}\text{d}\tau \\
& =\int_{s}^{s+\delta}\left\vert \text{Lie}(v)(\tau;t)\right\vert ^{p}\text{d
}\tau
\end{align*}
The timed Lie derivative satisfies
\begin{align*}
\mathcal{L}_{\tilde{f}}V(t,v) & =\left. \frac{d}{ds}\right\vert _{s=t}V(s,
\text{Lie}(v)(s,t)) \\
			      & =\left\vert \text{Lie}(v)(t+\delta;t)\right\vert
			      ^{p}-\left\vert v\right\vert ^{p} \\
			& \leq-(1-K^{p}e^{-p\lambda\delta})|v|^{p} \\
			&\leq-\frac{1-K^{p}e^{-p\lambda\delta}}{c_{2}}V(t,v).
\end{align*}
By choosing $\delta$ large enough such that $1-K^{p}e^{-p\lambda\delta}>0,$
we obtain (\ref{thmeq:1}) with $k=(1-K^{p}e^{-p\lambda\delta})/c_{2}$.
\end{proof}

\begin{remark}
As Theorem \ref{thm: Clift}, the above proof can be extended to Finsler
manifold, by replacing the Riemannian metric $g_{ij}$ by $\partial^{2}\left(
F^{2}\right) /\partial x_{i}\partial x_{j}$.
\end{remark}

\begin{remark}
	In \cite{andrieu2016transverse}, the authors obtained similar results of
Theorem \ref{thm: converse} in Euclidean space, see Proposition 1, 2, 3 \cite
{andrieu2016transverse}. More precisely, they proved the
equivalence of TULES-NL, UES-TL and ULMTE defined in
\cite{andrieu2016transverse}. We clarify their differences to our results.
First, in \cite{andrieu2016transverse}, the state space is $\mathbb{ R}^{n}$
with a metric described by positive definite matrices. Compared to Finsler
manifolds, it is easier to deal with and excludes some interesting examples, see
for example \cite{forni2013differential}. In contrast, Finsler structure is the
key object in this paper, it is more general and admits more complex structures.
More importantly, it helps us single out what are the more essential conditions
needed to guarantee contraction properties. For example, in
\cite{andrieu2016transverse}, it is required that the second order partial
derivatives of $f$ are uniformly bounded. On Finsler manifold, this condition is
no longer sufficient; instead, conditions imposed on the covariant derivative is
needed. See Definition \ref{def: Lip}. Second, the Lyapunov function constructed
in \cite {andrieu2016transverse} is quite different from the FLF constructed in
(\ref {cand FLF}). Thirdly, the proof of Theorem \ref{thm: converse} can be easily
extended to prove converse theorems of UGIAS systems.
\end{remark}

\begin{theorem}
\label{thm: converse UGIAS}Consider the system (\ref{NLTV no input}) defined
on Riemannian manifold $(\mathcal{X},g)$ with $f\in\mathcal{C}^{2}$ and
global Lipschitz continuous with constant $L$ in the sense of Definition \ref
{def: Lip}. If the system is UGIAS, and that the function $\beta(r,s)$ in
(\ref{eq:LIAS}) can be chosen such that \red{$\lim_{h \to 0+} \frac{\beta(h,t)}{h} =
\beta'(0,t)$ holds uniformly in 
$t\ge0$,} then there exists a (possibly time dependent) $\mathcal{C}^{0}$ FLF,
$V(t,v):\mathbb{R}_{+}\times T\mathcal{ X\rightarrow}\mathbb{R}$ such that:

\begin{enumerate}
\item There exist two class $\mathcal{K}_\infty$ functions $\gamma_{i}$, $i=1,2,$
such that
\begin{equation*}
\gamma_{1}\left( |y|\right) \leq V(t,v)\leq\gamma_{2}\left( |y|\right) ,\
\forall v\in T\mathcal{X}
\end{equation*}
where $|\cdot|$ is the norm induced by the Riemannian metric.

\item The timed Lie derivative of $V(t,v)$ along the system (\ref{CLift sys}
) satisfies
\begin{equation*}
\mathcal{L}_{\tilde{f}}V(t,v)\leq-\gamma_{3}(V(t,v)) 
\end{equation*}
for some class $\mathcal{K}_\infty$ functions $\gamma_{3}$.
\end{enumerate}
\end{theorem}

\begin{proof}
To streamline the proof, we prove the theorem in Euclidean space.
Thanks to Lemma \ref{lem: CLift}, in Euclidean space, the complete
lifted system is
\begin{equation} \label{euc:sys-clift}
\dot{x}    =f(x,t), \; \dot{y}    =\frac{\partial f(x,t)}{\partial x}y.
\end{equation}
Denote the solution with initial time $t$ and initial state $(x,y)$ as
$(\phi_{X}^{T}(\tau;t,x),\phi_{Y}^{T}(\tau;t,x,y))^{T}$. It is well known that
$\frac{\partial\phi_{X}}{\partial x}(\tau;t,x)$ satisfies the matrix ODE,
\[
\dot{X}=\frac{\partial f(\tau,\phi_{X}(\tau;t,x))}{\partial x}X,\ X(t)=I
\]
i.e. $ 
\frac{\partial\phi_{X}(\tau;t,x)}{\partial x}=\Phi(\tau,t)$, where
$\Phi(\tau,t)$ is the transition matrix corresponding to
$\frac{dz}{d\tau}=\frac{\partial
f(\tau,\phi_{X}(\tau;t,x))}{\partial x}z$ and hence 
\[
	\phi_{Y}(\tau;t,x,y) = 
	\frac{\partial\phi_{X}(\tau;t,x)}{\partial x} y.
\]
By assumption, 
\begin{equation} \label{euc:UGIAS}
|\phi_{X}(\tau;t,x_{1})-\phi_{X}(\tau;t,x_{2})|\leq\beta(|x_{1}-x_{2}
|,\tau-t),
\end{equation}
for all $x_{1},x_{2}\in\mathbb{R}^{n}$ and $\tau\geq t$ and a class
$\mathcal{KL}$ function $\beta$. 
We have the following estimations:
\begin{align*}
	\left|\frac{\partial \phi_X(\tau; t, x)}{\partial x_i}\right|
	&= \left|\lim_{h\ra
	0+}\frac{\phi_{X}(\tau;t,x+he_{i})-\phi_{X}(\tau;t,x)} {h} \right| \\
	& = \lim_{h\rightarrow 0+ }  \frac{|\phi_{X}(\tau;t,x+he_{i})-\phi_{X}(\tau;t,x)|}
{h} \\
	&\le \lim_{h\ra 0+} \frac{\beta(h,\tau-t)}{h} \\
	& = \beta'(0,\tau-t)
\end{align*} where $\beta'$ is the right-derivative of $\beta$ with respect to
the first argument, $e_i$ is the $i$-th component of the standard basis of
$\mathbb{R}^n$ and in the second equality we have used the fact that the
Euclidean norm is a continuous function so that $|\lim_{x\ra y} g(x)| = \lim_{x
\ra y} |g(x)|$.
Hence, there holds
\[
\left\vert \frac{\partial\phi_{X}(\tau;t,x)}{\partial x}\right\vert \leq
c\beta'(0,\tau-t)
\]
for all $x \in\mathbb{R}^{n}$ and $\tau\geq t$, where $c$ is a positive
constant. Consequently
\[
	|\phi_{Y}(\tau;t,x,y)| = \left\vert
	\frac{\partial\phi_{X}(\tau;t,x)}{\partial x} y\right\vert  \leq
	c\beta'(0,\tau-t)|y|,\ \forall
x,y\in\mathbb{R}^{n}.
\]

Let $\tilde{\beta}(r,s):=c\beta'(0,s)r$,
then $\tilde{\beta}$ is a class $\mathcal{KL}$ function since 1) $r \mapsto c
\beta'(0,s)r$ is class $\mathcal{K}$ and 2) $s \mapsto c\beta'(0,s)r$ is
decreasing for $\beta'$ is the derivative with respect to the first argument and
$\beta(r, \cdot)$ is decreasing for fixed $r$, \red{since by assumption
\begin{equation*}
\lim_{s\rightarrow\infty}\beta^{\prime}(0,s)=\lim_{s\rightarrow\infty}%
\lim_{h\rightarrow0+}\frac{\beta(h,s)}{h}=\lim_{h\rightarrow0+}\lim
_{s\rightarrow\infty}\frac{\beta(h,s)}{h}=0.
\end{equation*}}Now Proposition 7 \cite{sontag1998comments} implies the
existence of two class $\mathcal{K}_{\infty}$ functions $\alpha_{1},\alpha
_{2}$ such that
\[
\alpha_{1}\left(  \tilde{\beta}(r,s)\right)  \leq\alpha_{2}(r)e^{-s}.
\]
Define the candidate FLF as
\[
V(t,x,y)=\int_{t}^{\infty}\alpha_{1}\left(  |\phi_{Y}(\tau;t,x,y)|\right)
\text{d}\tau,
\]
which has the $\alpha_2(|y|)$ as upper bound:
\[
V(t,x,y)\leq\int_{t}^{\infty}\alpha_{2}(|y|)e^{-(\tau-t)}\text{d}
\tau=\alpha_{2}(|y|),
\] so $V$ is well-defined. It also has the lower bound (see \textbf{Step 1} in
the proof of Theorem \ref{thm: Clift}) :
\begin{align*}
	V(t,x,y) & \geq\int_{t}^{\infty}\alpha_{1}\left(  |y|e^{-L(\tau-t)}\right)
	\text{d}\tau \\
		 & =\int_{0}^{\infty}\alpha_{1}\left(  |y|e^{-Ls}\right)
		 \text{d}s \\
		 &:=\alpha_{3}(\left\vert y\right\vert ),
\end{align*}
where $\alpha_3$ is class $\mathcal{K}_\infty$ since
$\alpha_1$ is class $\mathcal{K}_\infty$. More precisely,
\begin{align*}
	\int_0^\infty \alpha_1(|y|e^{-Ls}) ds
	& \ge \int_0^1 \alpha_1(|y|e^{-Ls})ds  \\
	&\ge \int_0^1\alpha_1(|y|e^{-L})ds\\
	&= \alpha_1(|y|e^{-L}) \ra \infty
\end{align*} as $|y|\ra \infty$. Additionally, $\alpha_3$ being class
$\mathcal{K}$ is obvious. Similar to {\bf Step 2} above, we can calculate the
Lie derivative of $V(t,x,y)$ along the complete lift system
(\ref{euc:sys-clift}) as follows.
\begin{align*}
	\mathcal{L}_{\tilde{f}} &V(t, x,y) 
	= \left. \frac{\p}{\p \tau} \right|_{\tau = t}V(\tau; \phi_X(\tau;t,x),
		\phi_Y(\tau;t,x,y))\\
	&= \left. \frac{\p}{\p \tau} \right|_{\tau = t} \int_\tau^\infty
	\alpha_1(|\phi_Y(s;\tau,\phi_X(\tau;t,x), \phi_Y(\tau;t,x,y))|)ds \\
	&= \left. \frac{\p}{\p \tau} \right|_{\tau = t} \int_\tau^\infty 
	\alpha_1(|\phi_Y(s;t,x,y)|) ds \\
	&= -\alpha_1(|y|)
\end{align*}	
Summarizing,
\[
	\mathcal{L}_{\tilde{f}} V = -\alpha_1 (|y|) \le \alpha_1(
	\alpha_3^{-1}(V)), 
\] hence $V$ is indeed a FLF. Now letting $\gamma_1 = \alpha_3, \gamma_2 =
\alpha_2$ and $\gamma_3 = \alpha_1 \circ \alpha_3^{-1}$ will finish the proof.
\end{proof}

\begin{remark}
	The technical assumption of the differentiability of $\beta$ at
	$(0,s)$ is not very restrictive. It excludes only the case
	when the graph of $\beta(r,*)$ is tangent to the vertical axis at the origin.
	\red{However, uniformly differentiability of
	$\beta(\cdot, s)$ at the origin is an indeed strong assumption which
	makes this theorem less interesting compared to the integral form proved
	in \cite[Theorem 1]{angeli2002lyapunov}}.
	We do not know whether a smooth $\mathcal{KL}$
	function can be constructed when the system is UGIAS, even when
	the system is smooth. 
\end{remark}

\section{Rediscovery and Extension of Krasovskii's method}

When a UGIES system has $x_{\ast}$ as an equilibrium point, i.e. $
f(x_{\ast },t)=0$. It is obvious that the system is exponentially stable. The
converse Lyapunov theorem (see e.g. \cite{khalil2002nonlinear}) tells us
that there should exist a Lyapunov function $W(t,x)$ (not a FLF) for the
system (\ref{NLTV no input}) along which, the time derivative of the
Lyapunov function is negative definite. Now, having the UGIES property at
hand, by Theorem \ref{thm: converse}, a FLF can be constructed. A natural
question is, can we construct a Lyapunov function based on the information
of this FLF? The following proposition gives a rather interesting answer. 
As we will see, it is a rediscovery and extension of the
classical Krasovskii's method used for the construction of Lyapunov function 
\cite{khalil2002nonlinear}.

\begin{theorem}
\label{thm: kraso}Suppose the system $\dot{x}=f(x,t)$ is UGIES with a FLF $
V(t,v)$ with respect to the system (\ref{NLTV no input}) and have an
equilibrium point $x_{\ast}$. Then the system is UGES. Given a smooth time
invariant vector field $h$ on $\mathcal{X}$. If $[f,h]=0$, and 
\begin{equation*}
k_{1}d(x,x_{\ast})^{q}\leq\left\vert h(x)\right\vert \leq k_{2}d(x,x_{\ast
})^{q}
\end{equation*}
for $q\geq1$, then the function $W(t,x)=V(t,h)$ (or $V(t,x,h(x))$) is a
Lyapunov function for the system.
\end{theorem}

We need the following lemma to prove the theorem, which is interesting in
its own right.

\begin{lemma}
	Consider the system $\dot{x}=f(x,t)$ to which the solution is
	denoted as $
\phi(t;t_{0},x_{0})$. If there exists a vector field $h(x)$ on $\mathcal{X}$
such that $[f,h]=0$, then
$h(\phi(t;t_{0},x_{0}))\in T_{\phi(t;t_{0},x_{0})}\mathcal{X},\ \forall t\geq
t_{0}$
is the unique solution to the system (\ref{CLift sys}) with initial
condition $(t_{0};x_{0},h(x_{0}))$. In particular, the solution to (\ref
{NLTV no input}) started from $(x_{0},f(t_{0},x_{0}))$ is col$(\phi
(t;t_{0},x_{0}),f(\phi(t;t_{0},x_{0}),t_{0}))$.
\end{lemma}

\begin{proof}
The Lie bracket of $f$ and $h$ can be calculated as
\begin{equation*}
\lbrack f,h](\phi(t;t_{0},x_{0}))=\left. \frac{\text{d}(\phi^{\ast}h)}{\text{
d}s}\right\vert _{s=t}(\phi(t;t_{0},x_{0}))=0. 
\end{equation*}
Thus $
\phi^{\ast}h(\phi(t;t_{0},x_{0}))=\text{constant}=(x_{0},h(x_{0})), $
or
\begin{equation*}
h(\phi(t;t_{0},x_{0}))  =\phi(t;t_{0},x_{0})_{\ast}(x_{0},h(x_{0}))
=\text{Lie}(h(x_{0}))(t,t_{0}).
\end{equation*}
which completes the first half of the lemma. Since $[f,f]=0$ is always true,
the last claim also follows.
\end{proof}

\begin{proof}[Proof of Theorem \protect\ref{thm: kraso}]
It can be readily checked that $W(t,x)$ is a positive definite Lyapunov
candidate. Using the above lemma, we have
\begin{align*}
\mathcal{L}_{f}W(t,x) & =
 \left. \frac{\text{d}}{\text{d}\tau}\right\vert _{\tau=t}V(\tau
,h(\phi(\tau;t,x)))  \\
		      & =\mathcal{L}_{\tilde{f}}V(t,h(x)) \\
& \leq-kV(t,h(x)) \\
& =-kW(t,x),
\end{align*}
showing that $W(t,x)$ is indeed a Lyapunov function.
\end{proof}

\begin{corollary}
\label{cor: kraso}Consider the system $\dot{x}=f(x,t)$, where $x\in \mathbb{R
}^{n}$, with $f(0,t)=0$. If the system is ULIES with a FLF $V(t,x,\delta x)$.
 Assume that there exists a smooth vector field $h(x)$ on $\mathbb{R}^{n}$
such that $[f,h]=0$, where $h=0$ if and only if $x=0$, then the function $
W(t,x)=V(t,x,h(x))$ is a Lyapunov function such that the system is
exponentially stable. In particular, $W(t,x)$ can be chosen as $V(t,x,f(x))$
when the system is time invariant.
\end{corollary}

\begin{proof}
The time derivative of $W(t,x)$ reads
\begin{align*}
 \dot{W}(t &,x)  =\dot{V}(t,x,h(x)) \\
& =\frac{\partial V(t,x,h(x))}{\partial t}+\frac{\partial V(t,x,h(x))}{
\partial x}f(x,t) \\
& \;\;+\frac{\partial V(t,x,h(x))}{\partial\delta x}\frac{\partial h(x)}{
\partial x}f(x,t) \\
& =\frac{\partial V(t,x,h(x))}{\partial t}+\frac{\partial V(t,x,h(x))}{
\partial x}f(x,t) \\
& \;\;+\frac{\partial V(t,x,h(x))}{\partial\delta x}\frac{\partial f(x,t)}{
\partial x}h(x) \\
& \leq-kV(t,x,h(x)) \\
&=-kW(t,x),
\end{align*}
where the third equality follows from the fact in Euclidean space,
\begin{equation*}
\lbrack f,h]=\frac{\partial f}{\partial x}h-\frac{\partial h}{\partial x}f. 
\end{equation*}
Thus we see the system is exponentially stable with Lyapunov function $W(t,x)
$.
\end{proof}

Theorem \ref{thm: kraso} recovers and extends the so called Krasovskii's
method \cite{khalil2002nonlinear}: if there exists two constant positive
definite matrices $P$ and $Q$ such that
\begin{equation}
P\frac{\partial f(x)}{\partial x}+\left[ \frac{\partial f(x)}{\partial x}
\right] ^{T} P \leq-Q,   \label{eq: Krasovskii}
\end{equation}
then $V(x)=f^{T}(x)Pf(x)$ can serve as a Lyapunov function for the system 
since $h$ can be taken as $f$. Clearly, if (\ref{eq:
Krasovskii}) is satisfied, $\delta^{T}xP\delta x$ is a FLF for the system.
Then the Krasovskii's method is a direct consequence of Corallary \ref{cor:
kraso}. We consider two examples.

\begin{example}
Consider the linear system $\dot{x}=Ax$. Suppose there exists a FLF $
V=\delta x^{T}P\delta x$, such that $A^{T}P+PA=-I.$
Then since $[Ax,x]=0$, Corallary \ref{cor: kraso} tells us that when
replacing $\delta x$ with $x$, $V$ becomes a Lyapunov function, i.e. $
W(x)=x^{T}Px$. Furthermore, $x^{T}B^{T}PBx$ is also a Lyapunov function as
long as $B$ commutes with $A$ since in this case $[Ax, Bx]= (BA- AB)x =0$.
\end{example}

\begin{example}
\label{exmp: matrix measure}We consider the case when the matrix measure of
the Jacobian $J(t,x)=\partial f(x,t)/\partial x$ is uniformly bounded. That
is
\begin{equation*}
\mu(J(t,x))\leq-c<0,\ \forall t\geq0,\forall x 
\end{equation*}
This is considered in for example \cite{aminzare2014contraction}.
The FLF can be chosen as $V(x,\delta x)=|\delta x|$, and the Lyapunov
function $W(t,x)=|f(x,t)|$.
Indeed, it can be readilty checked that
\begin{align*}
	\dot{W}(t &,x(t))=\lim_{h\rightarrow0+}\frac{|f(t,x+hf(x))|-|f(x,t)|}{ h}\\
	& =\lim_{h\rightarrow0+}\frac{1}{h}\left( \left\vert f(x,t)+h\frac{\partial
		f(t,\xi)}{\partial x}f(x,t)\right\vert -|f(x,t)|\right) \\
	& \leq\lim_{h\rightarrow0+}\frac{\left\vert I+hJ(t,\xi)\right\vert -1}{h}
		|f(x,t)| \\
	& =\mu(J(t,\xi))|f(x,t)| \\
	& \leq-cW(t,x(t)).
\end{align*}
We see that although $f$ is time dependent, $V(t,x,f(x,t))$ may also have
the possibility to be a Lyapunov funtion. This sugggests that other tools
are needed to analyze such situation.
\end{example}

\begin{remark}
We remark that the results obtained by F. Bullo \cite{bullo2019geometric}
and K. Kosaraju \cite{kosaraju2019differential} (when the input $u$ is $0$)
regarding Krasovskii's method are special cases of Corallary \ref{cor: kraso}
.
\end{remark}

\section{Conclusion\label{sec: conclusion}}

Based on the paper \cite{forni2013differential}, we have given further
geometric and Lyapunov characterizations of incremental stability by
studying the complete lift of the system. We have shown that contraction
analysis can be carried out in a coordinate free way. Two converse contraction
theorems on Finsler manifolds, namely, UIES (UIAS) implies the existence of a
FLF. This result also confirms the differential framework proposed by F. Forni
\textit{et al.} is appropriate for analyzing incremental stability. The third
contribution is the establishment of the connections between incremental
stability and stability (of an equilibrium), which rediscovers and extends the
classical Krasovskii's method for constructing Lyapunov functions. Further
research includes applications of the proposed theories to observer design on
manifolds.

\section{Acknowlegement}

We thank Dr. Antoine Chaillet, Dr. Romeo Ortega, Dr. Fulvio Forni and Dr.
John W. Simpson-Porco for fruitful discussions during the preparation of the
manuscript.

\bibliographystyle{IEEEtran}
\bibliography{IEEEabrv, GeometricControl}

\end{document}